\documentclass{article}
\usepackage[round]{natbib}

\newif\ifappendix
\appendixtrue

\newcommand{\refappendix}[1]{\ifappendix
 Appendix~\ref{#1}\xspace
\else
 the supplementary material\xspace
\fi}

\usepackage[utf8]{inputenc} % allow utf-8 input
\usepackage[T1]{fontenc}    % use 8-bit T1 fonts
\usepackage{hyperref}       % hyperlinks
\usepackage{url}            % simple URL typesetting
\usepackage{booktabs}       % professional-quality tables
\usepackage{nicefrac}       % compact symbols for 1/2, etc.
\usepackage{microtype}      % microtypography
\usepackage{xcolor}         % colors
\usepackage{natbib}
\usepackage{graphicx}

\usepackage{amsfonts,amsthm,amssymb}
\usepackage[vlined,linesnumbered,ruled,longend]{algorithm2e}
\usepackage{nicefrac}
\usepackage{amsmath}
\usepackage{xspace}             % \xspace
\usepackage{adjustbox}          % \adjustbox 
\usepackage{multicol}            
\usepackage{subcaption}
\usepackage[noend]{algorithmic}
\usepackage{multirow}
\usepackage{hhline}
\usepackage{wrapfig}
\usepackage{float}

\usepackage{cleveref}

\theoremstyle{plain}
\newtheorem{theorem}{Theorem}[section]
\newtheorem{proposition}[theorem]{Proposition}
\newtheorem{lemma}[theorem]{Lemma}

\theoremstyle{remark}

\theoremstyle{definition}
\newtheorem{definition}[theorem]{Definition}

\newcommand{\dist}{\text{dist}}
\newcommand{\eps}{\varepsilon}
\newcommand{\opt}{\text{OPT}}

\newcommand{\cost}{\ensuremath{\text{cost}}}
\newcommand{\Opt}{\ensuremath{\text{Opt}}}
\newcommand{\reassign}{\ensuremath{\text{reassign}}}
\newcommand \email [1]{\small{\url{#1}}}

\usepackage{microtype}
\usepackage{graphicx}

\newcommand{\tO}{\ensuremath{\tilde O}}
\newcommand{\lspp}{\textsc{LocalSearch++}}
\newcommand{\clspp}{\textsc{ConstrainedLocalSearch++}}
\newcommand{\seedingalg}{\textsc{Seeding}}
\newcommand{\mybigchoose}[2]{\ensuremath{\genfrac(){0pt}{0}{#1}{#2}}}

\usepackage{textcomp}

\usepackage{mathtools}

\title{A Scalable Algorithm for Individually Fair K-means Clustering}

\begin{document}
\author{
MohammadHossein Bateni \\ Google Research \\ \email{bateni@google.com}
\and
Vincent Cohen-Addad \\ Google Research \\ \email{cohenaddad@google.com}
\and
Alessandro Epasto \\ Google Research \\ \email{aepasto@google.com}
\and
Silvio Lattanzi \\ Google Research \\ \email{silviol@google.com}
}
\maketitle

\begin{abstract}
We present a scalable algorithm for the individually fair ($p$, $k$)-clustering problem 
introduced by \cite{JKL20} and \cite{MV20}. 
Given $n$ points $P$ in a metric space, 
let $\delta(x)$ for $x\in P$ be the radius of the smallest ball around $x$ 
containing at least $\nicefrac nk$ points.  
A clustering is then called \emph{individually fair} if
it has
centers within distance $\delta(x)$ of $x$ for each $x\in P$.
While good approximation algorithms are known for this problem no efficient practical algorithms with good theoretical guarantees have been presented. We design the first fast local-search algorithm that runs in $\tO(nk^2)$ time and obtains a bicriteria $(O(1), 6)$ approximation. 
Then we show empirically that not only is our algorithm much faster than prior work, but it also produces lower-cost solutions.

\end{abstract}

\section{Introduction}

The $(p, k)$-clustering problems (with $k$-means, $k$-median and $k$-center as special cases) are widely used in many unsupervised machine-learning tasks for
exploratory data analysis, representative selection, data summarization, outlier detection, social-network community detection and signal processing, e.g., ~\cite{L82,M67,CG13,KAM19,ZWZ07,BKZ15}.

With such ubiquity of applications, it is fundamental to design fair
algorithms for such problems. 
In this paper we focus on the notion of individually fair clustering~\cite{JKL20}, 
which \emph{combines}
the $\ell_p$ cost objective with a $k$-center-based concept of fairness:
\emph{A minimum level of treatment should be guaranteed for every data point}.
To better understand this formulation consider the case in which centers were chosen randomly.
In this case any subset of $\nicefrac nk$ points would be expected to include one center.
So each point desires to be assigned to a center among its $\nicefrac nk$ closest points. 
This notion can be captured by considering a different
radius $\delta(x)$ for each $x$ in the dataset and 
by adding the constraint that there should be a center 
within distance $\delta(x)$ for each $x$. Satisfying such constraints amounts to
(a special case of) the priority $k$-center problem~\cite{Plesnik87, AS17,bajpai2021revisiting}.

Shortly after \cite{JKL20} proposed this problem and presented
a $2$-approximation for it, \cite{MV20} generalized it to an optimization
setting where an $\ell_p$ norm cost function (such as $k$-means or $k$-median)
is optimized within the space of individually fair solutions.
In fact, they devise a local-search algorithm with bicriteria $(84, 7)$ approximation (for $p=1$, that is $k$-median).

In recent years, several attempts have been made to improve these results, theoretically and practically. \cite{CN21} uses LP rounding
to improve the guarantee to $(2^{p+2}, 8)$, i.e., cost guarantee of 16 for $k$-means and 8 for $k$-median. In simultaneous work, \cite{HBKB22} presents an SDP-based
algorithm (without performance or runtime analysis), 
and \cite{VY21} presents LP-based bicriteria guarantees $(16^p+\epsilon, 3)$ for any $p$ and $(7.081+\epsilon, 3)$ for $p=1$.

Three of the above---with the exception of \cite{VY21}---present experimental studies. However, a major limitation of these algorithms is their running times, having an exponent of at least $4$ for the number of points $n$, making them impractical for real-world datasets of interest.
Therefore, all prior experiments were run
on small samples of size at most $1000$. As we will see in our empirical studies, our work is the first to report experimental evaluation on $600\times$ larger datasets.

Another related line of work is that of~\cite{chhaya2022coresets} (see also~\cite{DBLP:conf/focs/BravermanCJKST022} for state-of-the-art bounds) that introduced core-set based algorithms for regression and clustering with individual fairness. This work 
has a running time of $O(nkd+k^8d^4+(k \log n)^4)$ for individually fair clustering. While this reduces the running dependency on $n$ from prior work, it scales with a large polynomial of $k$ and $d$. This makes the work impractical for large dimensional datasets and when many centers are required. As we will show in the paper, our work improves significantly  the dependency on $k,d$ and allows us to  report experiments with datasets with $5 \times$ more points, $10 \times$ more dimensions and $3 \times$ more clusters than all prior work. 

\textbf{Our contributions.}

As mentioned, the previous algorithms (and our new result above) have poor runtime guarantees,
making them impractical for real-world datasets of interest. 
To address this shortcoming, we design a fast local-search algorithm
using ideas from the algorithm above as well as from the fast $k$-means algorithm of \cite{lattanzi2019better}. For simplicity, we focus on $k$-means, which is more commonly used
in practice although we notice that our approach can be generalized to all $p$-norms.

\begin{theorem}
\label{th:ls++}
There is an $\tO(nd+ nk^2)$-time algorithm for individually fair $k$-means
with a 6-approximation for radii and an $O(1)$-approximation on costs.
\end{theorem}

We complement our theoretical result with an experimental study.
In our experiments we use local search with swaps of size one,
building on the work of~\cite{lattanzi2019better} to find good swaps quickly (see also Beretta~\cite{abs-2309-16384}).
Whereas our algorithm scale to large datasets with 600,000 points dataset in less than two hours,
prior methods can process with at most 4000--25000 points in one hour;\footnote{Notice that the LP- or SDP-based algorithms require $\Omega(n^2)$ space, so those algorithms cannot scale to 100s of thousand point datasets even with several hours of runtime, as they run out of memory.} see Figure~\ref{fig:sample_vs_time}.

It is also interesting to note that despite the worse approximation guarantees, we observe in Section~\ref{sec:exp} that this algorithm outperforms prior algorithms on cost and fairness objectives.\footnote{Note that our algorithm, and some of the prior work~\cite{MV20, CN21, VY21}, work for any vector $\delta$ of radius bounds so the results presented in the paper are more general than the classic individual fair clustering setting.}

\textbf{Additional Related work.}
The $k$-means and $k$-median problems are NP-Hard, even in Euclidean space where they are hard to approximate within 
a factor 1.36 and 1.08 respectively (\cite{Cohen-AddadS19,Cohen-AddadSL22,DBLP:conf/soda/Cohen-AddadSL21} while the state-of-the-art approximation
algorithms achieve a 5.94 and 2.40 approximation respectively~\cite{Cohen-AddadEMN22}.
There is an extensive literature on group fairness, where the goal is
(in essence) to curb under- and over-representation in certain slices of the data
(say, sensitive groups based on gender or age group)~\cite{CKLV17, RS18, BCFN19, ahmadian2019clustering,AEKM20,AEKKMMPVW20,hotegni2023approximation,gupta2023efficient,froese2022modification}.
Another line of work concerns two generalizations of $k$-clustering problems to $\ell_p$ norm and ordered median objectives~\cite{BSS18, CS18, CS19, KMZ19}: Create a (non-increasingly) sorted vector out of the distances of points to their closest centers, and aim to minimize either the $\ell_p$ norm of this vector or the inner product of the vector with some given weight vector $w$. Note that $p=1, 2, \log n$ yields $k$-median, $k$-means and $k$-center through the first generalization, whereas $w=(1, 1, \dots, 1)$ and $w=(1, 0, \dots, 0)$ yield $k$-median and $k$-center objectives through the latter. \cite{CMV22} combines the two generalizations into the notion of $(p, q)$-Fair Clustering problem, which is also a generalization of Socially Fair $k$-Median and $k$-Means~\cite{GSV21,ABV21,MV21}.
Some of the above results are also motivated from the standpoint of solution \emph{robustness}---the main motivation stated in~\cite{HBKB22}.
The widely popular $k$-means clustering implicitly assumes certain uniform Gaussian distributions for the data~\cite{RBBL16}, and is known to be sensitive to sampling biases and outliers~\cite{WFKLSW20}. Beyond enforcing individual fairness or cluster-level consistency constraints (the focus of this work), researchers have tackled the above problem from various angles such as resorting to kernel methods~\cite{DGK04}, adding regularization terms~\cite{G16}, and using trimming functions~\cite{G16,DKP20,DKA21}.

\section{Preliminaries}
\label{sec:prel}
Let $(X, \dist)$ be a metric space, where $X$ is a finite set of 
points and $\dist$ a distance function between the points in $X$.
We define the distance between a point $p$ and a finite set of points $C$ as $\dist(p,C)=\min_{c\in C} \dist(p,c)$; if the set $C$ is empty we define the distance to be $\infty$. Let $\Delta=\max_{p,q}\dist(p,q) / \min_{p\neq q}\dist(p,q)$ denote the \emph{aspect ratio} of the instance.
We also let $\mu(X)$ denote the mean of a finite point set $X$.

\textbf{Problem definition.} 
Given a metric space $(X, \dist)$, 
the input to our problem is a tuple $(A, C, k, \delta)$, where
$A \subseteq X$ is a set of points of the metric space, $C \subseteq X$ is a set of points of the metric space, $k$ is a 
positive integer and $\delta : A \mapsto \mathbb{R}_+$ is the desired serving cost or \emph{radius} of points.
The goal is to output a set $S \subseteq C$ that minimizes 
$\sum_{a \in A} \dist(a, S)$ under the constraints that $|S| \le k$ and $\forall a \in A: \dist(a, S) \le \delta(a)$. 

The elements of a solution $S \subseteq C$ are called \emph{centers} or \emph{facilities}. Given a set $S$ of $k$ centers, let $\cost(S)_1$ denote the $k$-median cost of the set $C$ for the centers $S$, i.e.,
$\cost(S)_1 = \sum_{c \in A} \dist(c, S)$. Similarly, we define the cost of the set $C$ for the centers $S$ for the $k$-means problem as $\cost(S)_2 = \sum_{c \in A} \dist(c, S)^2$. 
In both setting we denote with $\Opt_k$ the cost of an optimum solution $S^*$. When it is clear from the context we will drop the index $1$ or $2$ from the notation for the cost.

A solution is an $(\alpha, \beta)$ bicriteria approximation if the $k$-median (or $k$-means) cost of the solution $S$ is at most $\alpha$ times that of the optimum, while the constraint that each $a$ should be at distance at 
most $\delta(a)$ from a center of the solution $S$ is violated by a factor at most $\beta$, i.e.: $\forall a$, 
$\dist(a, S) \le \beta \delta(a)$.

\section{Fast algorithm}
\label{sec:lspp}

In this section we focus on the $k$-means problem and we show how to modify the local-search algorithm presented in~\cite{lattanzi2019better} to obtain a bicriteria approximation for our problem, \Cref{th:ls++}. The key intuition is to use the concept of anchor zones introduced below to allow only
the swaps that preserve our fairness guarantees.
To fit the page limit, the proof of the lemmas of this section have been postponed to~\refappendix{appx:lspp}.

\begin{algorithm}[ht]
\caption{Scalable algorithm for individually fair $k$-means}
\label{alg:cls}
\begin{algorithmic}[1]
\REQUIRE{an input point set $X$, a desired number of cluster $k$, a target number of iterations $Z$, 
an approximation parameter $\gamma$}
\STATE{$C\gets\emptyset$, $S_0\gets\emptyset$}
\STATE{$S_0 \gets \seedingalg(X, \delta(\cdot),  \gamma)$} 
\STATE{\label{line:anchpoints} Define each point $p\in S_0$ as an \emph{anchor point}, and the ball $B(p)$ of radius $\gamma\delta(p)$ around $p$ as an \emph{anchor zone}.}
\STATE{Let $T \subseteq X\setminus S_0$ be a set of $k-|S_0|$ randomly selected points}
\STATE{$S \gets S_0 \cup T$}
\FOR{$i\gets2, 3, \dots, Z$}
    \STATE{$S\gets\clspp(X,S,B(\cdot))$}
\ENDFOR
\STATE{\textbf{return} $S$}
\end{algorithmic}
\end{algorithm}

\begin{algorithm}[ht]
\caption{\clspp}
\label{alg:cls++}
\begin{algorithmic}[1]
\REQUIRE{$X$, $S$, $S_0$, $B(\cdot)$}
\STATE{Sample $p\in X$ with probability $\frac{\cost(\{p\},S)}{\sum_{q \in X}\cost(\{q\},S)}$}
\STATE{$Q \gets \{q\in S \vert \forall x\in S_0: (S\setminus\{q\}\cup\{p\}) \cap B(x) \neq\emptyset\}$}
\STATE{$q^* \gets \arg\min_{q\in Q}\cost(X, S\setminus\{q\}\cup\{p\})$}
\IF{$\cost(X, S\setminus\{q^*\}\cup\{p\}) < \cost(X, S)$}
 	\STATE{$S \gets S\setminus\{q^*\}\cup\{p\}$}
\ENDIF
\STATE{\textbf{return} $S$}
\end{algorithmic}
\end{algorithm}

\begin{algorithm}[ht]
%\begin{algorithm}[t]
\caption{\seedingalg}
\label{alg:seeding}
\begin{algorithmic}[1]
\REQUIRE{$A$, $\delta(\cdot)$, $\gamma$}
\STATE{$S=\emptyset$}
\WHILE{$\exists p \in A : \dist(p,S) > \gamma \delta(p)$}
    \STATE{ $p^* \gets \arg\min_{p' \in \{p \in A \mid 
    \dist(p, S) > \gamma \delta(p)\}} \delta(p')$}
    \STATE{$S \gets S \cup \{p^*\}$}
\ENDWHILE
\STATE{\textbf{Output} $S$}
\end{algorithmic}
\end{algorithm}

For simplicity of exposition in this section we consider the classic setting where $A=C=X$, we refer the interested
reader to~\refappendix{sec:moretheory} for how the result can be extended to the more general case
where $A,C \subseteq X$.

Toward this end, we need to change both the initialization and the swapping procedure of the local-search algorithm to take into account the \emph{radius constraints}.
As for initialization %\todo{Is this the same as [MV20]?}
we first add a new center as long as there exists a point $p$ at distance greater than $\gamma\delta(p)$ from the current set of centers, namely we use Algorithm~\ref{alg:seeding} (whose correctness is proven in~\refappendix{appx:seeding}). We refer to the obtained set of centers as $S_0$. If $|S_0|$ is larger than $k$, then we know that the input is infeasible; otherwise we add additional points as centers until we obtain a set of $k$ centers $S$. We say that a point is an \emph{anchor point} if it is in $S_0$. Furthermore we define the ball $B(p)$ of radius $\gamma\delta(p)$ centered at $p$  as the \emph{anchor zone} for $p$.

As for the swaps, we choose a random point $q$ using $D^2$-sampling as in~\cite{lattanzi2019better}. If there is a subset $S'$ obtained by swapping an element of $S$ with $q$, such that (i) $|S'|=k$, (ii) every anchor zone contains at least one point in $S'$, and (iii) $\sum_{p \in X} \dist(p, S)^2 > \sum_{p \in X} \dist(p, S')^2$, then we change our current solution from $S$ to $S'$. 
Interestingly we show that after $O(k \log n\Delta)$ iterations, the solution will have constant-factor expected approximation for cost and moreover it violates the radius constraints by at most a factor of $2\gamma$. See the pseudocode in \cref{alg:cls}.

Now we show that our algorithm obtains a constant bicriteria approximation for individually fair $k$-means. 

Our proof uses many ingredients of  the proof in~\cite{lattanzi2019better} with careful modifications to handle the additional constraints imposed by the algorithm. In the remaining part of this section we prove our main theorem focusing on the novel part of our proofs. 

\subsection{Analysis (Proof of \Cref{th:ls++})}

As in~\cite{lattanzi2019better}, the main observation behind our proof is that every step of our algorithm in expectation reduces the solution cost by a factor $O\left(1-\frac{1}{k}\right)$.
In the following, given an input set of points $X$ containing at least 2 distinct points we will let $\Delta = \max_{p,q \in X} \dist(p,q) / \min_{p,q~:~p \neq q} \dist(p,q)$. 

Considering that the cost of the initial solution is at most $\Delta^2 n$, this implies that $O(k\log n\Delta)$ iterations suffice to obtain a constant approximation.

To simplify the exposition we assume that every cluster in the optimal solution has non-zero cost.\footnote{Note that this is w.l.o.g., since we can artificially increase the cost of every cluster by adding for each point a copy at distance
$\min_{p,q \in X} \dist(p,q)/n$.}

Next we state two lemmas outlining the algorithm's analysis that are central in our proof of \Cref{th:ls++}.
\Cref{th:ls++} itself is proven in~\refappendix{appx:th:ls++}.

\begin{lemma}\label{lem:main}
Let $X$ be the set of points from a feasible instance, $\gamma\geq 3$, and $S$ a set of centers with cost $\cost(X,S)>2000 \Opt_k $. With probability $\frac{1}{1000}$, for $S'=\clspp(X,S)$, we have $\cost(X,S')\leq (1-\frac{1}{100 k}) \cost(X,S)$.
\end{lemma}

\begin{lemma}\label{lem:k-means-cost}
Let $X$ be the  set of points from a feasible instance, and $\hat{S}$ 

a set of centers with $\cost(X,\hat S) \leq \gamma n\Delta^2$. After running $Z\geq 200000k\log (\gamma n\Delta)$ rounds of Algorithm~\ref{alg:cls++} on $\hat S$ outputs a solution $S$ such that $E[\cost(X, S)] \in O(\Opt_k)$.
\end{lemma}

\subsection{Proof of \Cref{lem:main}}
The proof in this section follows the structure of the proofs in~\cite{lattanzi2019better} with some fundamental modification to carefully handle the anchor zones constraints.

Before proving the lemma we recall two well-known results.
The following lemma is folklore:
\begin{lemma}
  \label{lemma:magicformula}
  Let $X \subseteq \mathbb{R}^d$ be a set of points and let $c \in \mathbb{R}^d$ be a center. Then we have
  $
  \cost(X,\{c\}) = |X| \cdot \|c-\mu(X)\|^2 + \cost(X,\mu(X)).
  $
\end{lemma}
We will also use the following lemma (rephrased from Corollary 21 in~\cite{DBLP:journals/corr/abs-1807-04518}).

\begin{lemma}
 \label{lemma:distances}
 Let $\epsilon >0$.
  Let $p,q \in \mathbb{R}^d$ and let $C \subseteq \mathbb{R}^d$ be a set of $k$ centers. Then
  $
  |\cost(\{p\},C) - \cost(\{q\},C)| \le \epsilon \cdot \cost(\{p\},C) + (1+\frac{1}{\epsilon}) \|p-q\|^2.
  $
\end{lemma}

We  assume that the optimal solution $S^* = \{c_1^*,\dots, c_k^*\}$ is unique (this can be enforced using proper tie breaking) and
use $X_1^*,\dots,X_k^*$ to denote the corresponding optimal partition. We will also use $S=\{c_1,\dots,c_k\}$ to refer to our current
clustering with corresponding partitions $X_1,\dots,X_k$. When the indices are not relevant, we will drop the index and write, for example, $c\in S$.

We use a notation and proof strategy similar to~\cite{KanungoMNPSW02}. 
We start by partitioning the optimal centers into \emph{anchor centers}, $A^*$, and \emph{unconstrained centers}, $U^*$. 

An optimal center is in $A^*$ if it is the closest optimal center to an anchor point (breaking ties arbitrarily), the remaining centers form the set of unconstrained centers. We say that an optimal center $c^*\in U^*$ is \emph{captured} by a center $c\in S$ if
$c$ is the nearest center to $c^*$ among all centers in $S$. Also we say that an optimal center $c^*\in A^*$ with corresponding anchor point $a$ is captured by a center $c\in S$ if
$c$ is the nearest center to $c^*$ among all centers in the anchor zone defined by $a$. Note that a center $c\in S$ may capture more than one optimal center and every optimal center is captured by
exactly one center from $S$ (ties are broken arbitrarily). Some centers in $c$ may not capture any optimal center. Similarly to \cite{KanungoMNPSW02} we call
these centers \emph{lonely} and we denote them with $L$. Finally, let $H$ be the index set of centers capturing exactly one cluster.
W.l.o.g., we assume that for $h\in H$ we have that $c_h \in S$ captures $c^*_h \in S^*$, i.e., the indices of the clusters with a one-to-one correspondence are identical.

Note that the above definition is slightly different from the classic definition in~\cite{KanungoMNPSW02}. In fact, an optimal center may not be captured by its closest center but by its closest center in the anchor zone. Nevertheless we can show that it is still possible to recover a similar result to the one in~\cite{lattanzi2019better} in this setting. 
Note one useful proposition of our definition, whose proof is deferred to the
appendix.
\begin{proposition}
 \label{proposition:dist}
 Let $c^*\in A^*$ be an optimal center with corresponding anchor point $a$, and let $c'$ be the closest point in $S$ to $c^*$,  and let $c^*$ be captured by the center $c\in S$. Then $\dist(c^*, c) \leq (\nicefrac{(\gamma+1)}{(\gamma-1)})\dist(c^*, c')$
\end{proposition}

We will use the above definition as in~\cite{lattanzi2019better}. Intuitively, if a center $c$ captures exactly one cluster of the optimal solution, we think of it as a candidate center for this cluster.
In this case, if $c$ is far away from the center of this optimal cluster, we argue that with good probability we sample a point close to the center. In order to analyze the change
of cost, we will argue that we can assign all points in the cluster of $c$ that \emph{are not in the captured optimal cluster} to a different center without increasing their contribution
by too much. This will be called the reassignment cost and is formally defined in the definition below. We will show that with good probability we sample from a cluster such that
the improvement for the points in the optimal cluster is significantly bigger than the reassignment cost.

If a center is lonely, we think of it as a center that can be moved to a different cluster. Again, we will argue that with high probability we can sample points from other clusters such that the reassignment cost is much smaller than the improvement for this cluster. 

Now we start to analyze the cost of reassignment of the points due to a center swap.
We would like to argue that the cost of reassigning the points currently assigned to a cluster center with index from $H$ or $L$ to other clusters is small. As discussed above, for $h \in H$,
we will assign all points from $X_h$ that are not in $X_h^*$ to other centers. For $l\in L$ we will consider the cost of assigning all points in $X_i$ to other
clusters. We use the following definition to capture the cost of this reassignment introduced in~\cite{lattanzi2019better}.

\begin{definition}
  Let $X\subseteq \mathbb{R}^d$ be a point set and $S \subseteq \mathbb{R}^d$ be a set of $k$ cluster centers and let $H$ be the subset of indices of cluster centers from $S = \{c_1,\dots,c_k\}$
  that capture exactly one cluster center of an optimal solution $S^* = \{c_1^*,\dots,c_k^*\}$. Let $X_i, X_i^*$, $1\le i\le k$, be the corresponding clusters.
  Let $h\in H$ be an index with cluster $X_h$ and w.l.o.g.\ let $X_h^*$ be the cluster in the optimal solution captured by $c_h$. The reassignment cost of $c_h$ is defined as

  $$\reassign(X,S,c_h) = \cost(X \setminus X_{h}^*, S\setminus \{c_h\}) -$$
  $$\cost(X \setminus X_h^*,S).$$

  For $\ell \in L$ we define the reassignment cost of $c_{\ell}$ as
   \begin{eqnarray*}
  \reassign(X,S,c_{\ell}) &=& \cost(X,S \setminus\{c_\ell\}) -\cost(X,S).
   \end{eqnarray*}
\end{definition}

We will now prove the following bound on reassignment costs.
We note that this proof is similar to the proof in~\cite{lattanzi2019better} but it includes key differences to handle the fact that optimal centers may not be assigned to the closest center in the current solution. The 
proof is provided in~\refappendix{appx:reassign}.

\begin{lemma}
\label{lemma:reassign}
For $r \in H \cup L$ we have
$$
\reassign(X,S,c_r) \le \frac{13}{100} \cost(X_r,S) + 332 \cost(X_r, S^*).
$$
\end{lemma}

Now that we have a good bound on the reassignment cost we make a case distinction. Recall that we assume that for every $h\in H$ the optimal center captured by $c_h$ is
$c_H^*$, i.e., the indices are identical. We first consider the case that $\sum_{h\in H}\cost(X_h^*,S)> \frac{1}{3} \cost(X,S)$.

With the previous lemma at hand, we can focus on the centers $h \in H$ where replacing $h$ by an arbitrary point close to the optimal cluster center of the optimal cluster captured by $h$
greatly improves solution cost. As in~\cite{lattanzi2019better} we call such clusters \emph{good} and make this notion precise in the following definition.

\begin{definition}
  A cluster index $h\in H$ is called \emph{good} if
 \begin{eqnarray*}
  \cost(X_h^*,S) - \reassign(X,S,c_h) - 9 \cost(X_h^*, \{c_h^*\}) > \\
  \frac{1}{100k} \cdot \cost(X,S).
  \end{eqnarray*}
\end{definition}

The above definition estimates the gain of replacing $c_h$ by a point close to the center of $X_h^*$ by considering a clustering that reassigns the points in $X_h$ that do not
belong to $X_h^*$ and assigns all points in $X_h^*$ to the new center. Now we want to show that we have a good probability to sample a good cluster. In particular,  we first argue that the sum of the cost of good clusters is large. We note that the following proof is a simple adaptation of~\cite{lattanzi2019better}; its proof is deferred to the
appendix.

\begin{lemma}
  \label{lemma:good}
  If $3 \sum_{h\in H} \cost(X_h^*,S) > \cost(X,S) \ge 2000 \Opt_k$, then
  $$
  \sum_{h\in H, h \text{ is good}}\cost(X_h^*,S) \ge \frac{9}{400} \cost(X,S).
  $$
\end{lemma}

 Now we  present a lemma from~\cite{lattanzi2019better} that whenever a cluster has high cost w.r.t.~$C$, it suffices to consider the points close to the optimal center
 to get an approximation of cluster cost. We will then use this fact to argue that we sample with good probability a point close to the center.

  \begin{lemma}[Lemma 6 from~\cite{lattanzi2019better} restated]
    \label{lemma:centralpoints}
  Let $Q \subseteq \mathbb{R}^d$ be a point set and let $S \subseteq \mathbb{R}^d$ be a set of $k$ centers and let $\alpha \ge 9$.
  If $\cost(Q,S) \ge \alpha \cdot \cost(Q,\{\mu(Q)\})$ then
  $$
  \cost(R,S) \ge \left(\frac{\alpha-1}{8}\right) \cdot \cost(Q,\{\mu(Q)\}),
  $$
  where $R\subseteq Q$ is the subset of $Q$ at squared distance at most $\frac{2}{|Q|} \cdot \cost(Q,\{\mu(Q)\})$ from $\mu(Q)$.
\end{lemma}

Now we can argue that sampling according to the sum of squared distances will provide us with constant probability with a good center.
Consider any index $h\in H$ with $h$ being good. We will apply \cref{lemma:centralpoints} with $Q=X_h^*$ and $\alpha = \cost(Q,S)/\cost(Q,\mu(Q))$.
Note that by the definition of good, we have that $\alpha \ge 9$. Now let us define $R_h^*$ to be the set $R$ guaranteed by \cref{lemma:centralpoints}.
We have $\cost(R_h^*,S) \ge \frac{\alpha -1}{8} \cost(X_h^*, \{c_h^*\}) = \frac{\alpha-1}{8 \alpha} \cost(X_h^*,S)\ge \frac{1}{9} \cost(X_h^*,S)$ by our choice of $\alpha$
(observe that $c_h^*$ equals $\mu(X_h^*)$). Since the sum of squared distances of points in good clusters 
is at least $9/400 \cost(X,S)$ by \cref{lemma:good}, we conclude that $\sum_{h\in H, h \text{ is good}} \cost(R_h^*, S) \ge \frac{9}{9 \cdot 400} \cost(X,S)$.
Thus, the probability to sample a point from $\sum_{h\in H, h \text{ is good}} \cost(R_h^*,S)$ is more than $1/400$. By the definition of good, if we sample such a point $c \in R_h^*$
we can swap it with $c_h$ to get a new clustering of cost at most $\cost(X,S \setminus \{c_h\} \cup \{c\}) \le \cost(X,S) - \cost(X_h^*,S) + \reassign(X,S,\{c_h\}) + \cost(X_h^*,\{c\})$.
By \cref{lemma:magicformula} we know that $\cost(X_h^*,\{c\}) \le 9 \cost(X_h^*, \{c_h^*\})$. Hence, with probability at least $1/400$ the new clustering has cost at most
\begin{eqnarray*}
&&\cost(X,S) - ( \cost(X_h^*,S) - \reassign(X,H,c_h)\\
&& - 9 \cost(X_h^*,\{c_h^*\})  \le (1-\frac{1}{100k}) \cdot \cost(X,S).
\end{eqnarray*}

To check that the swap is feasible we only need to make sure that the swap is feasible if $\{c_h^*\}\in A^*$. Otherwise we already know that the anchor balls are covered by other centers. If $\{c_h^*\}\in A^*$, let $a$ be the anchor point corresponding to $c_h^*$. Note that from the definition of good cluster, $\cost(X_h^*,S) - 9 \cost(X_h^*, \{c_h^*\}) > 0$ so by \cref{lemma:magicformula} we have $d(\{c_h^*\}, S)\geq 9 d(\{c_h^*\}, c)$. So given that the radius of the anchor ball is $3\delta(a)$ and the distance between $c_h^*$ and $a$ is bounded by $\delta(a)$ by triangle inequality we have that $c$ is inside the anchor ball.
This proves our lemma in the first case.

In the second case, we have  $\sum_{h\in H} \cost(X_h^*,S) < 1/3 \cost(X,S)$. Now let $R= \{1,\dots,k\} \setminus H$,
so we get $\sum_{r\in R} \cost(X_r^*,S) \ge 2/3 \cost(X,S)$. 
Observe that $R$ equals the index set of optimal
cluster centers that were captured by centers that capture more than one optimal center. 
This is because every optimal center is captured by one center and
$R$ does not include $H$. In this case, if the index of a center of our current solution is in $R\setminus L$ we cannot easily move the cluster center without having
impact on other clusters. What we do instead is to use the centers in $L$ as candidate centers for a swap. Note that those swaps are always feasible because inside each anchor ball we also have a center not in $L$.
Similar to the case above we will argue that we can swap a center from $L$ with a point that is close to an optimal center of a cluster $X_r^*$ for some $r\in R$.

Recall that we have already bound the cost of reassigning a center in $L$ so we just need to argue that the probability of sampling a good center is high enough.

In particular, we  focus on the centers $r \in R$ and swap an arbitrary center $\ell \in L$
 with an arbitrary point close to one of the centers in $R$ to improve the cost of the solution.
 Slightly overloading notation, we call such cluster centers \emph{good} and make this notion precise in the following definition.

\begin{definition}
  A cluster with index $i \in \{1,\dots,k\}$ is called \emph{good}, if there exists a center $\ell \in L$ such that
  \begin{eqnarray*}
  \cost(X_i^*,S) - \reassign(X,S,\ell) - 9 \cost(X_i^*, \{c_i^*\}) >\\ \frac{1}{100k} \cdot \cost(X,S).
  \end{eqnarray*}
\end{definition}

The above definition estimates the cost of removing $\ell$  and inserting a new cluster center close to the center of $X_i^*$ by considering a
clustering that reassigns the points in $X_i^*$ and assigns all points in $X_i^*$ to the new center. In the following we will now argue that the sum of cost of good clusters is large, this will be useful to show that the probability of sampling such a cluster is high enough. The proof of the following lemma is deferred to the Appendix.

\begin{lemma}
\label{lemma:good2}
If $3 \sum_{h\in H} \cost(X_h^*,S) \le \cost(X,S)$ and $\cost(X,S) \ge 2000 \Opt_k$ we have
  $$
  \sum_{r\in R, r \text{ is good}}\cost(X_r^*,S) \ge \frac{1}{20} \cost(X,S).
  $$
\end{lemma}

Note that also in this case we can now argue similarly as in the other case that sampling according to sum of squared distances will provide us 
with constant probability with a good center using \cref{lemma:centralpoints}. In fact, since the sum of
squared distances of points in good centers is at least $1/20 \cost(X,S)$ by \cref{lemma:good2}, it follows together with \cref{lemma:centralpoints}
that we sample a point from a good cluster $X_r^*$ that is within distance two times the average cost of the cluster with probability $\frac{1}{1000}$.
By the definition of a good cluster, we know that such a point improves the cost of the current clustering by at least a factor of $(1-\frac{1}{100k})$. Thus, Lemma~\ref{lem:main} follows.

\section{Empirical analysis}
\label{sec:exp}

In this section we evaluate empirically the algorithms presented and we compare them with state-of-the-art methods from the literature.
In our analysis, all datasets used are {\em publicly available}. We implemented our algorithms, and the other baselines in Python, and we ran each instance of our experiments independently using a single-core from  machines within our institution's cloud with x86-64 architecture, 2.25GHz and using less than $32GB$ of RAM. To foster the reproducibility of our experiments we released our code open source.\footnote{Our code is available open source at the following link: \url{https://github.com/google-research/google-research/tree/master/individually_fair_clustering/}.}

\paragraph{Datasets.}
We used several real-world datasets from the UCI Repository~\cite{Dua:2017} and from the SNAP library, that are standard in the clustering literature. This includes: adult~\cite{kohavi1996scaling} $n=32561$, $d=6$,  bank~\cite{moro2014data} $n=45211$, $d=3$,  diabetes~\cite{Dua:2017} $n=101766$, $d=2$,
gowalla~\cite{cho2011friendship}, $n=100000$, $d=2$,
skin~\cite{bhatt2010skin}, $n=245057, d=4$,  shuttle\footnote{Thanks to NASA for releasing the dataset.}  $n=58000, d=9$, and covertype~\cite{blackard1999comparative},  $n=581012, d=54$. For consistency with prior work, for adult, bank and diabetes we use the same set of columns used in the analysis of~\cite{MV20}. We preprocess each dataset to have zero mean and unit standard deviation in every dimension. All experiments use the Euclidean distance.
The effect of the value of $k$ is discussed in~\refappendix{sec:effectofk}.

\paragraph{Algorithms.} 
We consider the following algorithms.\\
-- \textbf{VanillaKMeans}. Standard non-fair k-means implementation from \texttt{sklearn}. This baseline represents the k-means cost achievable neglecting fairness. \\
-- \textbf{ICML20}~\cite{MV20}:
We implemented the algorithm following the recommendation of the  paper (i.e., using a single swap in the local search and a factor $3\delta(p)$ instead of $6\delta(p)$ in the initialization). We set $\epsilon=0.01$ in the algorithm. \\
-- \textbf{NeurIPS21}~\cite{CN21}: 
We use the Python code provided by the authors.\footnote{The code was obtained from \url{https://github.com/moonin12/individually-fair-k-clustering} and adapted.} We use both the more accurate algorithm \textbf{NeurIPS21} and the faster algorithm using sparsification (\textbf{NeurIPS21Sparsify}).\\
-- \textbf{Greedy}: Similarly to prior work we consider the execution of the greedy seeding algorithm as a baseline.  \\
-- \textbf{LSPP}: We implemented our local-search algorithm with modifications similar to that of ICML20 (a single swap and $\mu=3$ factor in seeding algorithm). We also modified the algorithm to run only a fixed number of local-search iterations (namely $500$) in all experiments.

Moreover, we also design a fairness preserving Lloyd's method, the F-Lloyd's method, that we add as a post-processing of our algorithm. In the F-Lloyd's method we assign each point to the nearest center. Then we obtain the mean of the clusters. Notice that the mean minimizes the $k$-means cost, but it may not be a feasible solution for the distance bound. For this reason, we use the anchor points obtained by our algorithm to find  the next center  approximating, with binary search, 
the closest feasible point to the mean (respecting the anchor points constraints), on the line between the current center and the mean. Though this procedure does not alter the theoretical guarantees, it improves the results empirically. We use $20$ iterations of the F-Lloyd's method at the end of our algorithm.\footnote{Later, we discuss the applicability of this improvement to other prior baselines.} 
%\end{itemize}

\begin{figure*}
\centering
\subfloat[Time (secs)]{\label{fig:sample_vs_time}\includegraphics[width=55mm]{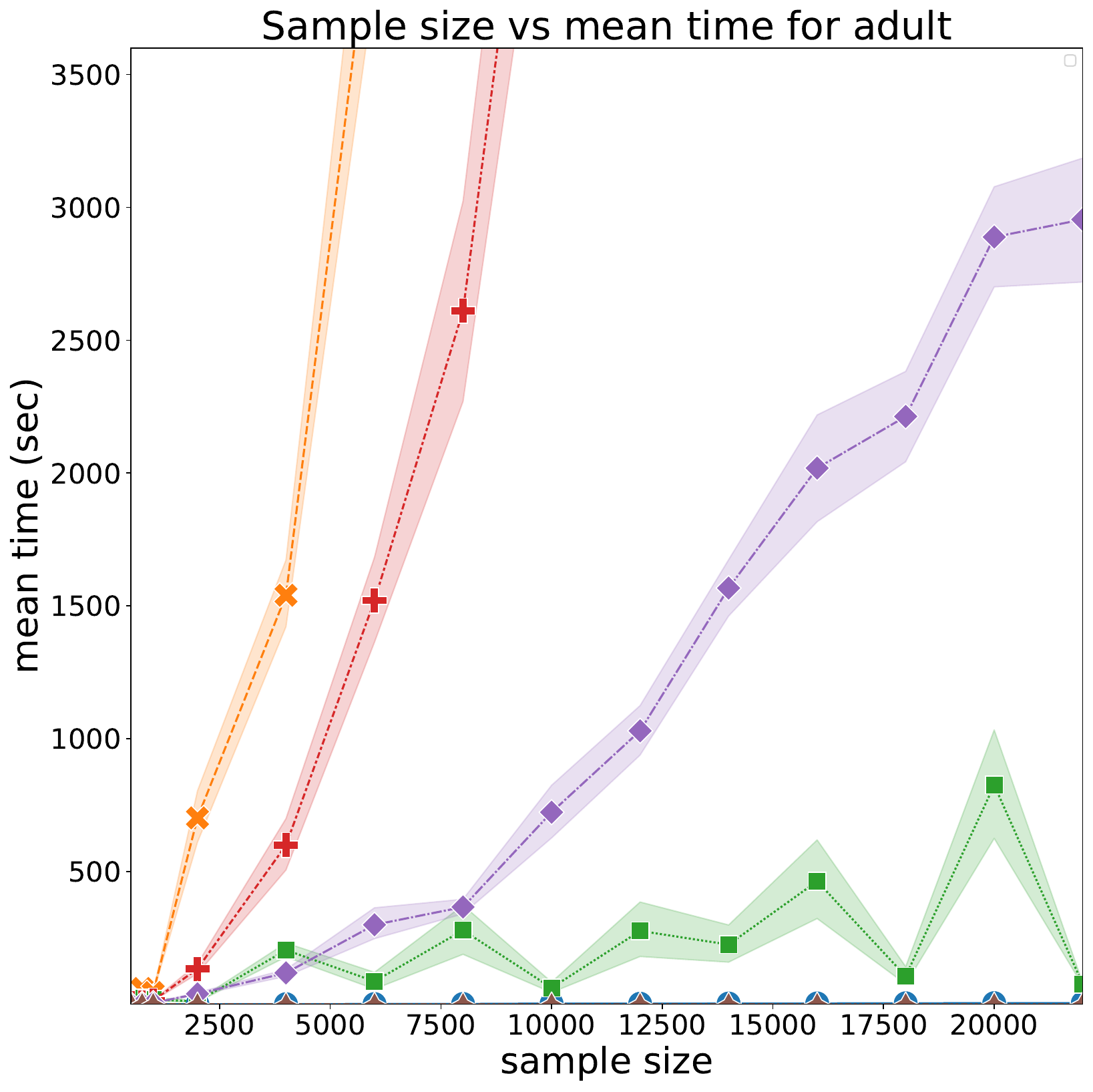}}
\hskip2mm
\subfloat[Cost]{\label{fig:sample_vs_cost}\includegraphics[width=55mm]{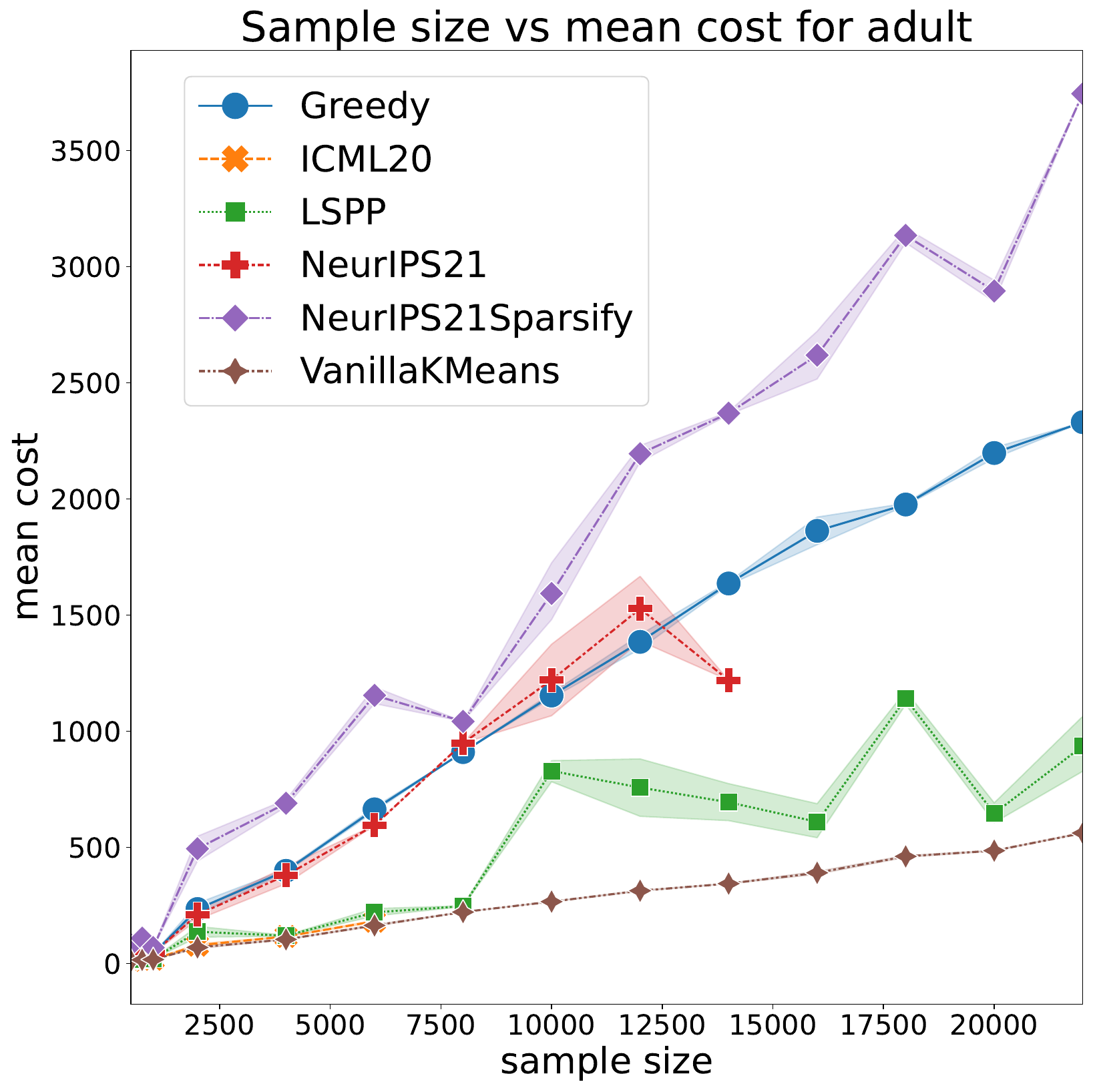}}
\hskip2mm
\subfloat[Bound ratio]{\label{fig:sample_vs_bound}\includegraphics[width=55mm]{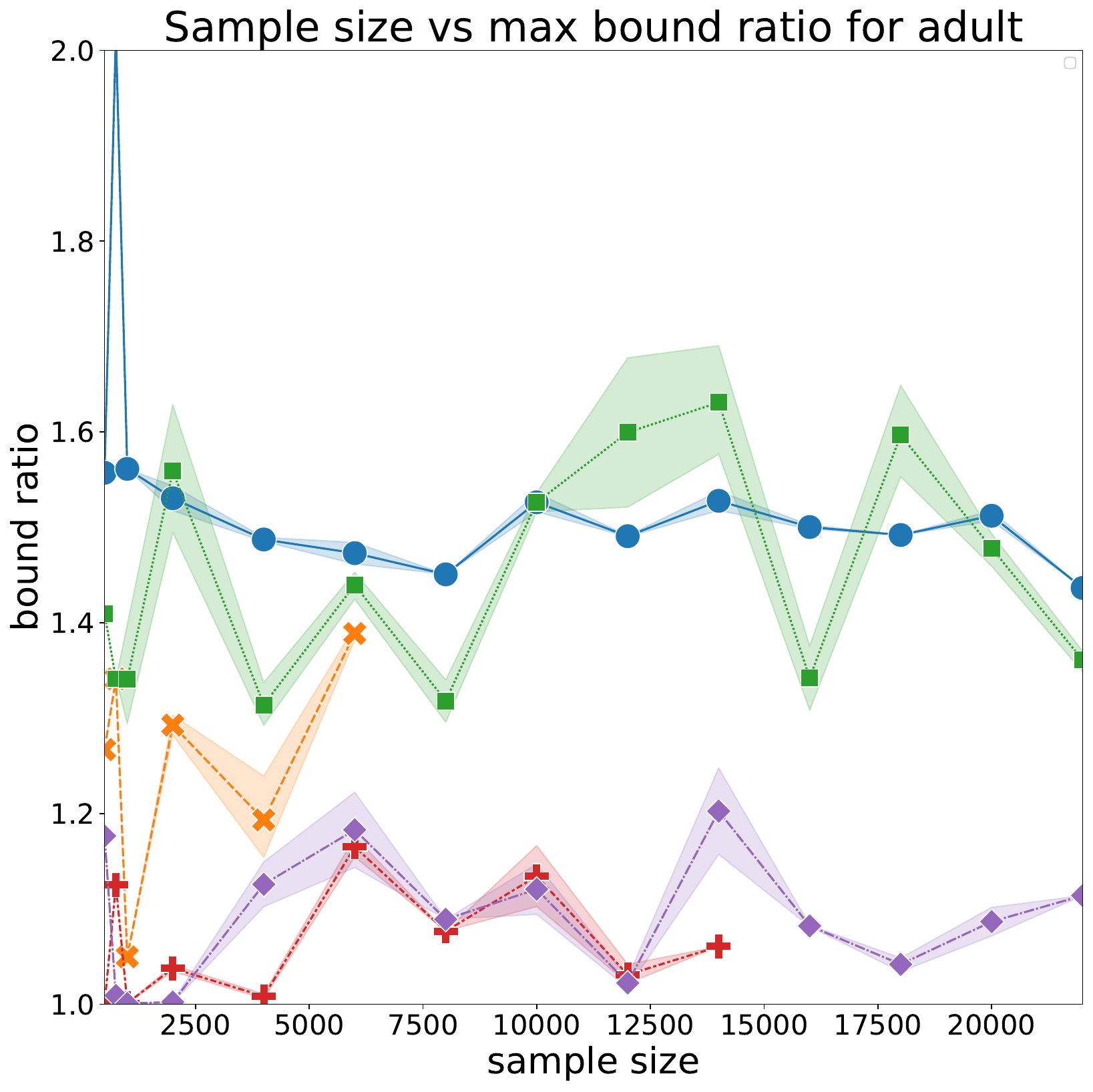}}
\caption{Mean completion time, cost, and bound ratio for the algorithms on Gowalla dataset subsampled to different sizes, $k=10$. The shades represent the $95\%$ confidence interval (notice that some algorithms are deterministic). Runs that did not complete in $1$ hour on the sample are not reported. VanillaKMeans bound ratio is > 60 in all runs and not show in the plot as out of scale).\label{fig:samples_vs_stats}}
\end{figure*}

\textbf{Metrics.}
We focus on three key metrics: the $k$-means cost of clustering, the average runtime of algorithm and the {\it bound ratio} 
$\max_{p}\frac{\dist(p,S)}{\delta(p)}$
where $S$ is the solution of the algorithm.
We repeat each experiment configuration $10$ times and report the mean and standard deviation of the metrics.

\textbf{Comparison with other baselines.}
In this section we report a comparison of our algorithm with the other baselines. For all experiments, unless otherwise specified, we replicate the setting of individual fairness~\cite{MV20} for $\delta(p)$, by setting $\delta(p)$ as the distance to the $\nicefrac nk$-th nearest point. 

Notice that the ICML20 algorithm evaluates, for each iteration of local search, all possible swaps of one center with a non-center while NeurIPS21 and NeurIPS21Sparsify both depend on computing all-pairs distances in $O(dn^2)$ time. This makes these algorithms not scalable to large datasets, unlike our algorithm. Therefore all prior experiments~\cite{MV20,CN21} used a subsample of $\approx\hskip-1mm 1000$ elements from the datasets to run their algorithm. In this section we use a similar approach for consistency.

In Figure~\ref{fig:sample_vs_time}, \ref{fig:sample_vs_cost}, \ref{fig:sample_vs_bound}, we report the results of the various algorithms for different sizes of the sample on the gowalla dataset, fixing $k=10$. 

The main message of the paper is summarized in Figure~\ref{fig:sample_vs_time}.
We allowed each algorithm to run for up to $1$ hour on increasingly large samples of the data. Notice how our algorithm is orders of magnitude faster than all fair baselines including the faster NeurIPS21Sparsify variant. ICML20 does not complete in $1$ hour past size $\sim 5000$, NeurIPS21 past $\sim 7000$,  NeurIPS21Sparsify past $\sim 25000$ while our algorithm's running time is not highly affected by scale. This will allow us to run on $500,000$ sized datasets, orders of magnitude larger than SOTA algorithms. 

Then we focus on the $k$-means cost of the solution in Figure~\ref{fig:sample_vs_cost}. Notice that our algorithm (LSPP) has a cost better (or comparable) than that of all fair baselines and close to the unfair VanillaKMeans.

Finally, Figure~\ref{fig:sample_vs_bound} shows the max ratio of %the
distance of a point $p$ to centers vs $\delta(p)$. We observe that VanillaKMeans has a bound ratio between $60$ and $90$ (out of the plot scale), confirming the need for fair algorithms. For the remainder of the paper we will ignore VanillaKMeans. 

Now we focus on the fair algorithms. Notice that  the (much slower) ICML20 has statistically comparable bounds with LSPP (and both significantly better than their worst case guarantees), while NeurIPS21 and NeurIPS21Sparsify %instead 
have slightly better bounds.

\paragraph{Experiments on the full datasets.}
The scalability of our algorithm allows us to run it on the full datasets with up to $\nicefrac12$ million elements. In this section, we run our algorithm and the fast Greedy baseline on all complete datasets, using $k=10$. Our algorithm completed all runs in less than $2$ hours on the full datasets.

To compare with all the slower baselines we allow ICML20, NeurIPS21 and NeurIPS21Sparsify  to run on a subsample of the data containing 4000 points (but we evaluate the solution on the entire dataset). This of course has no theoretical guarantee and can perform especially poorly in case of outliers. In the supplementary material we report the standard deviation for the metrics in Table~\ref{table:cost-bound-all-datasets-large}.

For this large-scale experiment, the input bound $\delta(p)$ for each point $p$ is set using the $\nicefrac nk$-th closest point in a random sample of 1000 elements. 

In all but one dataset, our algorithm has a significantly lower $k$-means cost than that of all other baselines. Similarly to the above results, our algorithm has a similar or better ratio bound than that of ICML20 (with the sampling heuristic), while the ratio bound of NeurIPS21 and NeurIPS21Sparsify is sometimes lower. In any instance our algorithm has a much better ratio that the worst-case theoretical guarantees.
The results are reported in Table~\ref{table:cost-bound-all-datasets-large}.

\begin{table}
\scriptsize
\begin{center}
\begin{tabular}{llcc}
\toprule
dataset & algorithm & $k$-means cost & bound ratio \\
\midrule
adult & Greedy &      1.56E+05 &           1.8 \\
     & ICML20 &      6.59E+04 &           1.4 \\
     & NeurIPS21 &      1.14E+05 &           \textbf{1.2} \\
     & NeurIPS21Sparsify &      1.02E+05 &           \textbf{1.2}\\
     & LSPP &      \textbf{6.14E+04} &           1.4 \\
     \cmidrule(lr{.75em}){1-4} 
bank & Greedy &      8.57E+04 &           1.9 \\
     & ICML20 &      3.23E+04 &           1.6 \\
       & NeurIPS21 &      5.68E+04 &           \textbf{1.2} \\
     & NeurIPS21Sparsify &      5.70E+04 &           \textbf{1.2}\\
   & LSPP &      \textbf{3.02E+04} &           1.6 \\
   \cmidrule(lr{.75em}){1-4} 
covtype & Greedy &      3.33E+07 &           1.3 \\
     & ICML20 &      2.84E+07 &          \textbf{ 1.1 }\\
     & NeurIPS21 &      2.76E+07 &           \textbf{1.1 }\\
     & NeurIPS21Sparsify &      2.80E+07 &          \textbf{ 1.1} \\
     & LSPP &     \textbf{ 2.50E+07} &           \textbf{1.1} \\
     \cmidrule(lr{.75em}){1-4} 
diabetes & Greedy &      6.60E+04 &           2.7 \\
     & ICML20 &      3.00E+04 &           1.3 \\
     & NeurIPS21 &      N/A &           N/A \\
         & NeurIPS21Sparsify &      3.36E+04 &           \textbf{1.2} \\
 & LSPP &      \textbf{2.66E+04} &           1.4 \\
 \cmidrule(lr{.75em}){1-4} 
gowalla & Greedy &      1.17E+04 &           1.6 \\
     & ICML20 &      N/A &           N/A \\

            & NeurIPS21 &      8.65E+03 &          \textbf{1.1} \\
            & NeurIPS21Sparsify &      1.63E+04 &           1.2 \\
   & LSPP &      \textbf{6.93E+03} &           1.6 \\
 \cmidrule(lr{.75em}){1-4} 
shuttle & Greedy &      4.89E+05 &           2.3 \\
     & ICML20 &      1.91E+05 &           2.0 \\
       & NeurIPS21 &      2.60E+05 &           \textbf{1.0} \\
     & NeurIPS21Sparsify &      2.72E+05 &           1.1 \\
   & LSPP &      \textbf{1.79E+05} &           2.1 \\
   \cmidrule(lr{.75em}){1-4} 
skin & Greedy &      1.80E+05 &           2.1 \\
     & ICML20 &      \textbf{7.47E+04} &           1.8 \\
     & NeurIPS21 &      9.36E+04 &           \textbf{1.1} \\
     & NeurIPS21Sparsify &      1.03E+05 &           \textbf{1.1} \\
     & LSPP &      9.27E+04 &           3.1 \\
\bottomrule 
\end{tabular}
\caption{\small Mean Cost and max bound ratio for all full-sized datasets and k=10 with ICML20, NeurIPS21 and NeurIPS21Sparsify ran on a sample of 4000 points.}
\label{table:cost-bound-all-datasets-large}
\vspace{-0.15in}
\end{center}
\end{table}

Finally, we focus on the impact of our novel F-Lloyd's fair improvement method. We observe that this step can only be applied to algorithms based on anchor points (like ours and ICML20) and that (of course) cannot improve the efficiency of the slow ICML20 baseline. For this reason we tested ICML20 with our F-Lloyd method on a sample of size $1000$ points where such an algorithm can run. The full analysis is reported in the supplemental material and confirms all prior experiments. We observe that adding F-Lloyd steps improves the cost of ICML20 baseline as well, but does not change the overall trend reported, i.e., that our algorithm has comparable (or better cost) than all fair baselines while scaling to orders of magnitude larger datasets. We believe that the F-Lloyd algorithm could benefit other work in the future.

\section{Conclusions and Future Works}
We present a new scalable algorithm for individually fair clusters, with strong
theoretical guarantees and good experimental performances.
An interesting open question is to use the more recent analysis of \lspp\ by~\cite{choo2020k} to improve the running time of our algorithm to $O(dnk^2)$. 
However,
it is not clear how to obtain the necessary strong guarantees similar to the one in Lemma~12 of~\cite{choo2020k}.

\bibliographystyle{apalike}
\bibliography{references}

\section{Additional experimental results}

\paragraph{Small scale datasets and effect of F-Lloyd improvement}

The results observed before are confirmed in all datasets, as shown in Table~\ref{table:cost-bound-all-datasets-1000-extended}, where we report the cost and bound ratio for all datasets, subsampled to 1000 elements, and $k=10$. 

In this experiment we report as well the results of using our novel F-Lloyd's fair improvement method on top of the ICML20 baseline (recall that this step can only be applied to algorithms based on anchor points  such as LSPP and ICML20). Even with this improvement, the picture remains unchanged LSPP has always the lowest cost (or close to lowest cost, achieved by the improved ICML20+F-Lloyd) to all fair algorithms while scaling to $2$ orders of magnitude larger datasets than ICML20 (even without the additional running time of F-Lloyd which only makes ICML20 slower). For this reason, for the remainder of the paper we focus on the ICML20 baseline as introduced by the authors.

\paragraph{Effect of $k$}
A similar overall picture appears in Figure~\ref{fig:k_vs_cost},~\ref{fig:k_vs_bound}, where we report the results of the various algorithms for different $k$'s on a sample of 1000 elements in the adult dataset.
Notice that our algorithm has lower or comparable cost to all fair baselines  in cost in Figure~\ref{fig:k_vs_cost} and the comparable or slightly higher bound ratio in Figure~\ref{fig:k_vs_bound}.

\label{sec:effectofk}
\begin{figure*}[ht]
%\begin{figure}
\centering
\subfloat[Cost]{\label{fig:k_vs_cost}\includegraphics[width=6cm]{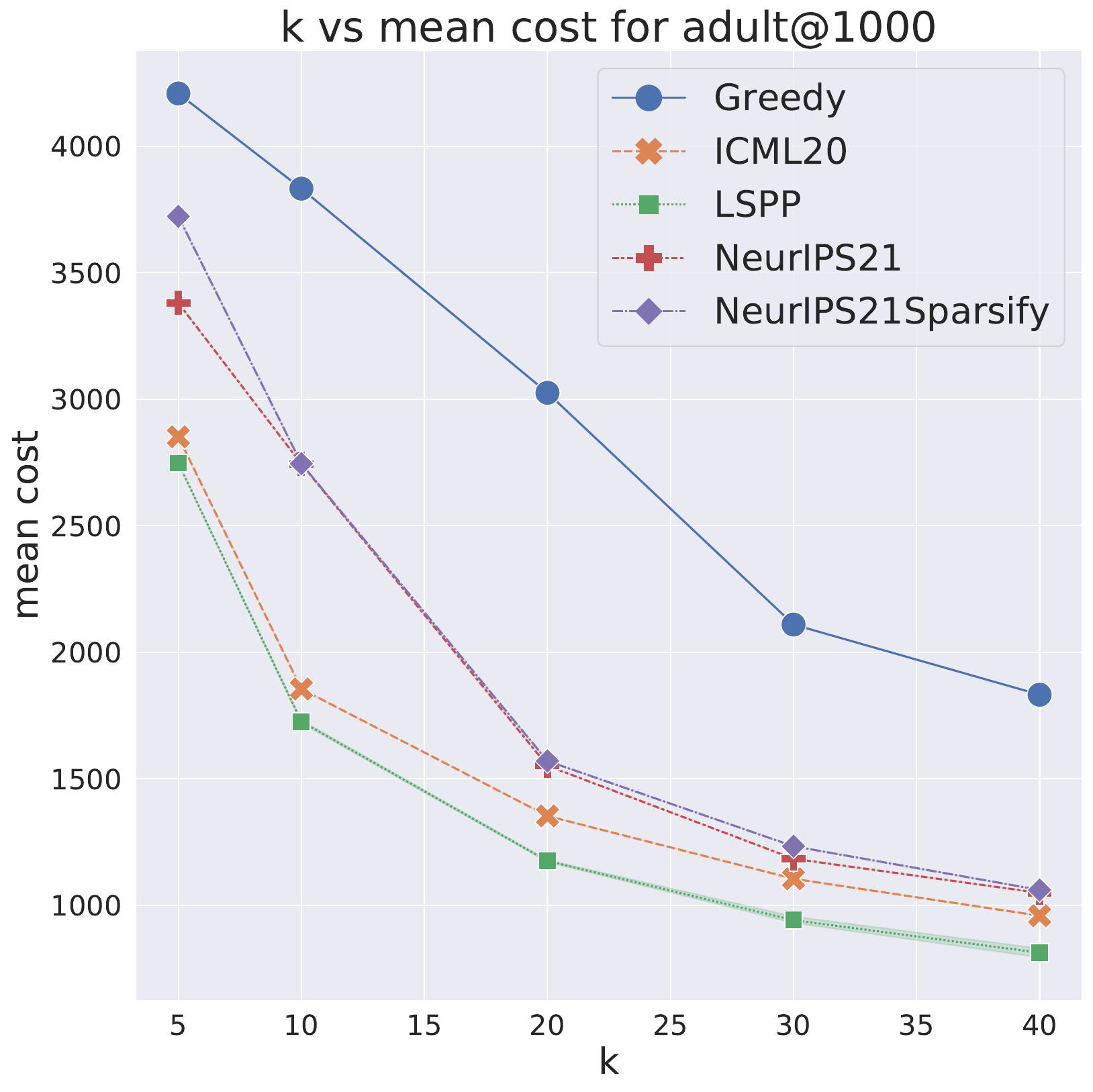}}
\subfloat[Bound ratio]{\label{fig:k_vs_bound}\includegraphics[width=6cm]{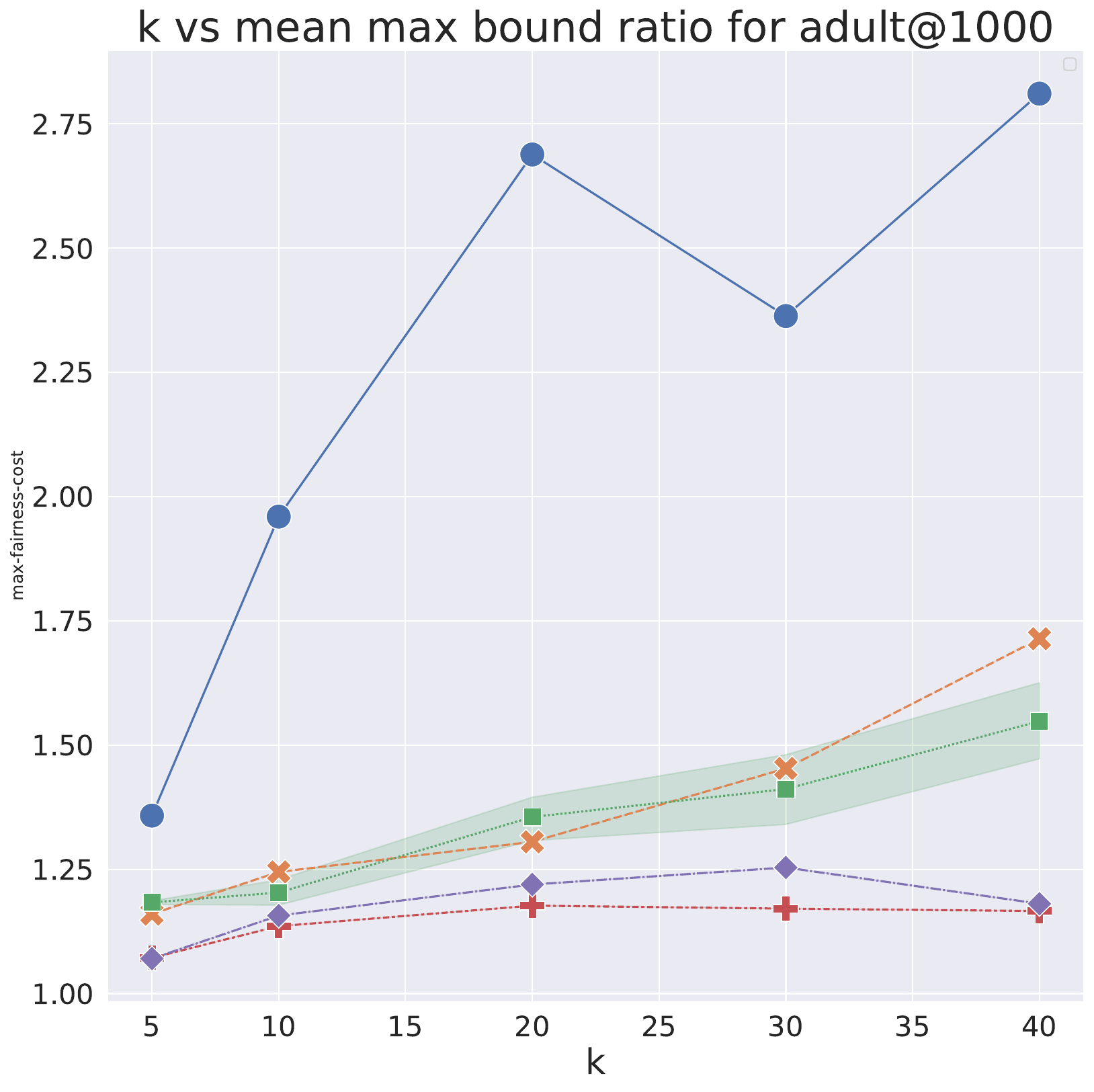}}
\caption{Mean completion cost and bound ratio for the algorithms on the adult dataset subsampled to 1000 elements and different $k$'s. The shades represent the $95\%$ confidence interval. \label{fig:k_vs_stats}}
\end{figure*}

\paragraph{Additional results on large-scale datasets}

As we mentioned before, our algorithm is the only one that runs on the big datasets within reasonable time (and memory).

To compare with all the slower baselines we allow ICML20, NeurIPS21 and NeurIPS21Sparsify  to run on a subsample of the data containing 4000 points (but we evaluate the solution on the entire dataset). This of course has no theoretical guarantee and can perform especially poorly in case of outliers.

For this large-scale experiment, the input bound $\delta(p)$ for each point $p$ is set using the $\nicefrac nk$-th closest point in a random sample of 1000 elements. 

The results in Table~\ref{table:cost-bound-all-datasets-large}
shows that in all but one dataset, our algorithm has a significantly lower $k$-means cost than that of all other baselines. Similarly to the above results, our algorithm has a similar or better ratio bound than that of ICML20 (with the sampling heuristic), while the ratio bound of NeurIPS21 and NeurIPS21Sparsify is sometimes lower. In any instance our algorithm has a much better ratio that the worst-case theoretical guarantees.

\begin{table*}
\small
\begin{center}
\begin{tabular}{llll}
\toprule
dataset & algorithm &  $k$-means cost &  bound ratio \\
\midrule
adult & Greedy &        3832.6 (---)&                2.0 (---) \\
     & ICML20 &        1854.9 (---)&                1.2 (---) \\
     & ICML20+F-Lloyd &        1733.2 (---)&                1.3 (---)\\
     & NeurIPS21 &        2744.3(---) &                \textbf{1.1} (---)\\
     & NeurIPS21Sparsify &        2745.5 (---)&                1.2 (---)\\
     & LSPP &        \textbf{1726.0} (9.4)&                1.2 (0)\\
     \cmidrule(lr{.75em}){1-4} 
bank & Greedy &        1081.8 (---) &                1.8 (---) \\
     & ICML20 &         568.4 (---) &                1.4 (---) \\
    & ICML20+F-Loyd &         517.9 (---) &                1.5 (---) \\

     & NeurIPS21 &         784.2 (---)&                \textbf{1.2} (---)\\
     & NeurIPS21Sparsify &         761.2 (---)&                \textbf{1.2} (---)\\
     & LSPP &          \textbf{510.5} (6.9)&                1.4 (0.1) \\
     \cmidrule(lr{.75em}){1-4} 
covtype & Greedy &       50629.8 (---) &                1.3 (---) \\
     & ICML20 &       42121.5 (---) &                1.1 (---) \\
     & ICML20+F-Lloyd &       \textbf{35211.3} (---) &                1.0 (---)\\
     & NeurIPS21 &       47810.6 (---)&                1.1(---) \\
     & NeurIPS21Sparsify &       46078.6 (---)&                1.1 (---)\\
     & LSPP &          35860.4 (569)&                \textbf{1.0} (0.0)\\
     \cmidrule(lr{.75em}){1-4} 
diabetes & Greedy &         522.3 (---) &                2.1 (---) \\
     & ICML20 &         266.4 (---) &                \textbf{1.1} (---) \\
    & ICML200+F-Loyd &         \textbf{231.8} (---) &                \textbf{1.1} (---) \\
& NeurIPS21 &         N/A &                N/A \\
     & NeurIPS21Sparsify &         299.8 (---)&               \textbf{1.1} (---)\\
     & LSPP &          248.5 (8.5)&                1.2 (0.2)\\

     \cmidrule(lr{.75em}){1-4} 
gowalla & Greedy &          41.4 (---)&                1.6 (---)\\
            & ICML20 &  
            18.6 (---)&                 \textbf{1.0} (---)\\
                & ICML200+F-Lloyd &         \textbf{17.3} (---) &                \textbf{1.0} (---) \\
& NeurIPS21 &          44.3 (---)&                 \textbf{1.0} (---)\\
            & NeurIPS21Sparsify &          68.6 (---) &               \textbf{1.0} (---)\\
    & LSPP &          20.3 (1.6) &                1.3 (0.1) \\

     \cmidrule(lr{.75em}){1-4} 
shuttle & Greedy &        2647.4 (---) &                2.0 (---) \\
     & ICML20 &        1335.0 (---) &                1.8 (---) \\
    & ICML200+F-Lloyd &         1272.1 (---) &                1.8 (---) \\
   & NeurIPS21 &        2494.8 (---)&                \textbf{1.0}(---)\\
   & NeurIPS21Sparsify &        2477.1 (---)&                1.2 (---)\\
    & LSPP &       \textbf{1219.1} (31.3)&                1.8 (0.1) \\
    \cmidrule(lr{.75em}){1-4} 
skin & Greedy &         584.3 (---) &                2.7 (---) \\
     & ICML20 &        292.8(---) &                  2.3 (---) \\
         & ICML200+F-Lloyd &         \textbf{276.2} (---) &                2.4 (---) \\
 & NeurIPS21 &         379.4 (---) &                 \textbf{1.1} (---) \\
     & NeurIPS21Sparsify &         384.1 (---) &                \textbf{1.1} (---) \\
    & LSPP &           314.3 (27.7)&                2.5 (0.1)\\
\bottomrule
\end{tabular}
\caption{\small Cost and max bound ratio for all datasets subsampled for 1000 elements and $k=10$ (stddev in parentheses for the LSPP randomized algorithm). N/A indicates that the algorithm did not complete in 2 hours. In this table we also report experiments with the F-Lloyd heuristic applied to ICML20.}
\label{table:cost-bound-all-datasets-1000-extended}
\end{center}
\end{table*}

\paragraph{Standard deviation of the metrics in large datasets.}
In Table~\ref{table:cost-bound-all-datasets-large-std} we report the standard deviation for the metrics in Table~\ref{table:cost-bound-all-datasets-large}. Notice that in this experiment, the input to ICML20, NeurIPS21 and NeurIPS21Sparsify algorithms are run on a random subsample, so this makes the algorithm non-deterministic. Notice that our algorithm has statistically significantly lower cost than the other baselines in almost all datasets.

\begin{table*}
\begin{center}
\begin{tabular}{llcc}
\toprule
dataset & algorithm & cost stddev & bound ratio stddev\\
\midrule
adult & ICML20 &      6.81E+02 &           1.00E-01 \\
     & NeurIPS21 &      1.34E+04 &           0.00E+00 \\
     & NeurIPS21Sparsify &      7.65E+03 &           0.00E+00 \\
     & LSPP &      8.84E+02 &           0.00E+00 \\
     \cmidrule(lr{.75em}){1-4} 
bank & ICML20 &      1.04E+03 &           1.00E-01 \\
     & NeurIPS21 &      5.95E+03 &           0.00E+00 \\
     & NeurIPS21Sparsify &      6.57E+03 &           0.00E+00 \\
     & LSPP &      6.85E+02 &           1.00E-01 \\
     \cmidrule(lr{.75em}){1-4} 
covtype & ICML20 &      2.23E+05 &           0.00E+00 \\
     & NeurIPS21 &      1.85E+05 &           0.00E+00 \\
     & NeurIPS21Sparsify &      3.83E+05 &           0.00E+00 \\
     & LSPP &      4.51E+05 &           0.00E+00 \\
     \cmidrule(lr{.75em}){1-4} 
diabetes & ICML20 &      8.19E+02 &           2.00E-01 \\
     & NeurIPS21 &      N/A&           N/A \\
     & NeurIPS21Sparsify &      1.22E+03 &           1.00E-01 \\
     & LSPP &      9.58E+02 &           1.00E-01
     \\
     \cmidrule(lr{.75em}){1-4} 
     gowalla & ICML20 &      N/A &           N/A \\
     
     & NeurIPS21 &      2.07E+03 &           0.00E+00 \\
            & NeurIPS21Sparsify &      4.05E+03 &           1.00E-01 \\
      &      LSPP &      1.32E+03 &           0.00E+00 \\
          \cmidrule(lr{.75em}){1-4} 
shuttle & ICML20 &      1.25E+04 &           2.00E-01 \\
     & NeurIPS21 &      6.57E+03 &           0.00E+00 \\
     & NeurIPS21Sparsify &      1.37E+04 &           0.00E+00 \\
     & LSPP &      1.05E+04 &           3.00E-01 \\
     \cmidrule(lr{.75em}){1-4} 
skin & ICML20 &      3.51E+03 &           3.00E-01 \\
     & NeurIPS21 &      1.60E+03 &           1.00E-01 \\
     & NeurIPS21Sparsify &      1.31E+04 &           0.00E+00 \\
     & LSPP &      4.80E+02 &           2.00E-01 \\
\bottomrule
\end{tabular}
\caption{\small Standard deviation of cost and max bound ratio for all full-sized datasets and k=10 with ICML20, NeurIPS21 and NeurIPS21Sparsify ran on a sample of 4000 points .}
\label{table:cost-bound-all-datasets-large-std}
\vspace{-0.15in}
\end{center}
\end{table*}

\paragraph{Effect of the normalization of the points.}
In our experiments, we applied two preprocessing steps that are common in the clustering literature~\cite{borassi2020sliding}:  each point in the dataset is shifted so that the dataset has zero mean; and each dimension is scaled to have unit standard deviation. We have observed that such pre-processing has no significant effect on the experimental conclusion of our work.  
In Table~\ref{table:cost-bound-all-datasets-large-not-norm} we report results for the same experiments previously reported in Table~\ref{table:cost-bound-all-datasets-large}, this time
 executed on the adult dataset without normalization.

As observed before, our algorithm has a significantly lower (or comparable) $k$-means cost than that of other baselines, better ratio than the worst-case theoretical guarantees and a much faster runtime than all fair algorithms. 

\begin{table*}[ht]
\begin{center}
\begin{tabular}{llrr}
\toprule
         &                   &  $k$-means-cost &  max-fairness-cost \\
dataset & algorithm &               &                    \\
\midrule
adult (no norm.) & Greedy &       5.0E+13 &            2.9E+00 \\
         & ICML20 &      \textbf{ 1.2E+13} &            \textbf{ 1.8E+00} \\
         & NeurIPS21 &       N/A &            N/A \\
         & NeurIPS21Sparsify &       5.2E+13 &            \textbf{ 1.8E+00} \\
         & LSPP &       1.3E+13 &            2.3E+00 \\
\bottomrule
\end{tabular}
\end{center}
\caption{Mean Cost and max bound ratio for all full-sized datasets and k=10 with ICML20, NeurIPS21 and NeurIPS21Sparsify ran on a sample of 4000 points on the Adult datasets without normalization.}
\label{table:cost-bound-all-datasets-large-not-norm}
\end{table*}

\section{Seeding Strategy for Local Search}
\label{appx:seeding}

We describe the seeding procedure outlined in \cref{alg:seeding} to initialize our local-search approach.
The proof of this lemma can be found in Section~\ref{appx:seeding}.

\begin{lemma}
\label{lem:seeding}
    If the problem is feasible, \Cref{alg:seeding} with parameter $\gamma > 2$ returns a set of points $S$ of 
    size at most $k$
    such that each point $p$ is at distance at most $\gamma \delta(p)$ from 
    the closest point in $S$, i.e., $\forall p\in A: \dist(p, S) \leq \gamma\delta(p)$.
\end{lemma}

\begin{proof}[Proof of Lemma~\ref{lem:seeding}]
The proof is similar to the proof of correctness of Gonzales' algorithm for $k$-center~\cite{gonzalez1985clustering} and
of Hochbaum and Shmoys' algorithm~\cite{DBLP:journals/mor/HochbaumS85}.
Observe first that by feasibility, there cannot be $k+1$ points $p'_1, \ldots, p'_{k+1}$ such that the balls $\delta(p'_i)$ are all pairwise disjoint (since otherwise the optimum solution would need $k+1$ centers to satisfy the $\delta(p')$ constraints).

Let $p_1,\ldots,p_{k^*}$ be the sequence of points picked by the algorithm. We have that $\delta(p_i) \le \delta(p_j)$ for any 
$i \le j$. Note that at the end of the algorithm, each point $p$ is at distance at most $\gamma \delta(p)$ from one of $p_1,\ldots,p_{k^*}$ so what remains to be shown is that $k^* \le k$. We claim that the collection of balls centered at the $p_i$ and of radii $\delta(p_i)$ are all pairwise disjoint and so if the problem is feasible, the algorithm does not return more than $k$ points (i.e.: $k^* \le k$). Consider a pair $i,j$ and without loss of generality $i > j$. We have that $p_i$ is at distance at least $\gamma  \delta(p_i)$ from $p_j$ by the definition of the algorithm. Since $\gamma > 2$ and $\delta(p_i) \ge \delta(p_j)$, we have that $\delta(p_i) + \delta(p_j) \le 2 \delta(p_i) < \gamma \delta(p_i) \le \dist(p_i,p_j)$ and so the ball of radius $\delta(p_j)$ around $p_j$ cannot intersect the ball of radius $\delta(p_i)$ around $p_i$.

\end{proof}

\section{Proof of Section~\ref{sec:lspp}}
\label{appx:lspp}
The proof in this section follows closely the structure of the proofs in~\cite{lattanzi2019better} with some modification to carefully handle the anchor zones constraints.

\subsection{Proof of \Cref{th:ls++}}
\label{appx:th:ls++}
\begin{proof}[Proof of \Cref{th:ls++}]
 Observe that the dimension can be reduced to 
$O(\log k /\eps^{-2})$ using the Johnson-Lindenstrauss transform~\cite{BecchettiBC0S19,MakarychevMR19}. 
Hence, in $\tilde{O}(nd)$ time one can find a projection to a space 
of dimension $O(\log k)$ that preserves the $k$-means cost of all solutions up to an $O(1)$ factor and execute the 
algorithm in this space. The claimed running time then follows immediately.

The algorithm returns infeasible only if
it finds $k+1$ disjoint individual fairness balls. 
But in that case,
the problem is infeasible (their fairness constraints cannot be satisfied with $k$ points).

Let $\hat{S}$ be the set $S$ before calling \clspp.
In this set, every point $p$ has distance at most $\gamma \delta(p)$ from a center so $\cost(X,\hat{S})\leq \gamma n \Delta^2$.
\cref{lem:k-means-cost}
then shows
that after $Z$ calls to \clspp, we obtain a constant approximation.

Now we show that at any point in time during the execution of the algorithm $\max_{p \in X} \dist(p, S) \le 2\gamma\delta(p)$.
The
algorithm guarantees to keep at least one point in every anchor ball. Moreover every point $p$ is at distance at most $\gamma\delta(p)$ from an anchor point $p'$ with $\delta(p)>\delta(p')$. The anchor ball $B(p')$ must have a center $c\in S$, so $\dist(c,p')\leq \gamma\delta(p')$. Thus by triangle inequality $\dist(c,p)\leq 2\gamma \delta(p)$.

It takes $O(dnk)$ time to compute the initial set $\hat{S}$.
To implement the local search, we need to compute the cost of swapping the new sample point with an old center. This requires iterating over all clusters and for each cluster we need to compute the distance to all other centers and to check that there is at least one center in each anchor ball. Thus, a local search step requires $O(dkn + dk)$ time in the worst case,
resulting in a total runtime of $O(dnkZ)$. The Theorem follows.
\end{proof}

\subsection{Proof of Lemma~\ref{lem:k-means-cost}}
\begin{proof}
By Lemma~\ref{lem:main} we know that if before any call of \clspp\ the cost of the centers is bigger than
$2000 \Opt_k $ then with probability $\frac{1}{1000}$ we reduce the cost by a $(1-\frac{1}{100 k})$ multiplicative factor.

Now consider another random process $Y$ with initial value equal to $\cost(X,\hat{S})$, which for $Z =100000k\log n\Delta(X)^2$ iterations, it reduces the value by a $\left(1-\frac{1}{100k}\right)$ multiplicative factor with probability $1/1000$, and finally increases the value by 
an additive $2000 \Opt_k $. It is not hard to see that the final value of $Y$ stochastically dominates the cost of the solution our algorithm produces.

So the final expected value of $Y$ is larger than the expected value of 
$\cost(X, S)$ conditioned on the initial clustering $\hat{S}$. Furthermore,
\begin{align*}
E[Y] &= 2000 \Opt_k + \cost(X,\hat{S}) \cdot
 \sum_{i=0}^{Z} \mybigchoose{Z}{i} \left(\frac{1}{1000}\right)^i \left(\frac{999}{1000}\right)^{Z - i} \left(1-\frac{1}{100k}\right)^i \\
 &= \cost(X,\hat{S})\left(1-\frac{1}{100000k}\right)^Z 
 + 2000 \Opt_k \\
 &\leq  \frac{\cost(X,\hat{S})}{n\Delta(X)^2} + 2000\Opt_k.
\end{align*}

This implies that $E[\cost(X, S)| \hat{S}] \leq\frac{\cost(X,\hat{S})}{n\Delta(X)^2} + 2000 \Opt_k$. Our upper-bound
on the cost of $\hat{S}$ is deterministic, hence $E[\cost(X, S)] \leq\frac{\cost(X,\hat{S})}{n\Delta(X)^2} + 2000 \Opt_k \leq 2001 \Opt_k$.
\end{proof}

This section contains the proofs of the lemmas of Section~\ref{sec:lspp}.

\subsection{Proof of Lemma~\ref{lem:k-means-cost}}
The proof in this section follows closely the structure of the proofs in~\cite{lattanzi2019better} and we include it for completeness.
\begin{proof}[Proof of Lemma~\ref{lem:k-means-cost}]
By Lemma~\ref{lem:main} we know that if before any call of \clspp\ the cost of the centers is bigger than
$2000 \Opt_k $ then with probability $\frac{1}{1000}$ we reduce the cost by a $(1-\frac{1}{100 k})$ multiplicative factor.

We next use another random process $Y$ to handle dependencies between rounds. In this way we can have a coupling with an independent process that is easier to analyze.
We let $Y$ be a random process with initial value equal to $\cost(X,\hat{S})$, which for $Z =100000k\log n\Delta^2$ iterations, it reduces the value by a $\left(1-\frac{1}{100k}\right)$ multiplicative factor with probability $1/1000$, and finally increases the value by 
an additive $2000 \Opt_k $. It is not hard to see that the final value of $Y$ stochastically dominates the cost of the solution our algorithm produces.
So the final expected value of $Y$ is larger than the expected value of 
$\cost(X, S)$ conditioned on the initial clustering $\hat{S}$. Furthermore,
\begin{align*}
E[Y] &= 2000 \Opt_k + \cost(X,\hat{S}) \cdot
 \sum_{i=0}^{Z} \mybigchoose{Z}{i} \left(\frac{1}{1000}\right)^i \left(\frac{999}{1000}\right)^{Z - i} \left(1-\frac{1}{100k}\right)^i \\
%\end{eqnarray*}
%\begin{eqnarray*}
 &= \cost(X,\hat{S})\left(1-\frac{1}{100000k}\right)^Z + 2000 \Opt_k \\
%  \end{eqnarray*}
% \begin{eqnarray*}
 &\leq  \frac{\cost(X,\hat{S})}{n\Delta^2} + 2000\Opt_k.
\end{align*}

This implies that $E[\cost(X, S)| \hat{S}] \leq\frac{\cost(X,\hat{S})}{n\Delta^2} + 2000 \Opt_k$. Our upper-bound
on the cost of $\hat{S}$ is deterministic, hence $E[\cost(X, S)] \leq\frac{\cost(X,\hat{S})}{n\Delta^2} + 2000 \Opt_k \leq 2001 \Opt_k$.
\end{proof}

\subsection{Proof of~\Cref{proposition:dist}}

\begin{proof}[Proof of~\Cref{proposition:dist}]
  If $c'$ is within distance $\gamma \delta(a)$ to $a$, the lemma follows from $\dist(c^*, c) = \dist(c^*, c')$ by definition of $c$ and anchor ball. Otherwise we know that $c^*$ is at distance at most $\delta$ from $a$, $c$ is at distance at most $\gamma\delta$ from $a$, and $c'$ is  at distance at least $\gamma\delta$ from $a$. The lemma follows the triangle inequality.
\end{proof}

\subsection{Proof of Lemma~\ref{lemma:reassign}}
\label{appx:reassign}
\begin{proof}[Proof of Lemma~\ref{lemma:reassign}]

We only present the case $r \in H$. The case $r\in L$ is almost identical (in fact, simpler).
We observe that $\reassign(X,S,c_r) = \cost(X_r \setminus X_r^*, S \setminus \{r\}) - \cost(X_r\setminus X_r^*, S)$ since vertices in clusters
other than $X_r$ will still be assigned to their current center.
If $r\in H$, we assign every point in $X_r \cap X_i^*$, $i\not= r$, to the center that captured the center of $X_i^*$.
While this assignment may not be optimal, its cost provides an upper bound on the cost of reassigning the points:
We move every point in $X_r \cap X_i^*$, $i\not= r$, to the center of $X_i^*$.
Now the closest center of $S$ to these points is a center with distance close to the one that captured the center of
$X_i^*$, which, for points not in $X_r^*$,
cannot be $r$, since $r$ is in $H$. 
The fact that the squared moved distance of each point equals its contribution to the optimal solution allows us to get an upper bound on the cost change using \cref{lemma:distances}.
After this, we move the points back to their original location 
while keeping their cluster assignments fixed. 
Again we can use the bound on the overall moved distance together 
with \cref{lemma:distances} to obtain a bound on the change of cost. 
Combining the two gives an upper bound 
on the increase of cost that comes from reassigning the points. Details follow.

Let $Q_r$ be the (multi)set of points obtained from $X_r \setminus X_r^*$ by moving each point in $X_i^*\cap X_r$, $i\not= r$, to $c_i^*$.
We  apply \cref{lemma:distances} with $\epsilon = 1/100$ to get an upper bound for the change of cost with respect to $S$ that results from moving
the points to $Q_r$. For $p\in X_r \setminus X_r^*$ let $q_p\in Q_r$ be the point of $Q_r$ to which $p$ has been moved. We have: 
  \begin{eqnarray*}
     &&  |\cost(\{p\},S) - \cost(\{q_p\},S)| 
     \le \frac{1}{100} \cost(\{p\},S) + 101 \cdot \cost(\{p\},S^*).
  \end{eqnarray*}
Summing up over all points in $X_r \setminus X_r^*$ yields
  \begin{eqnarray*}
     && |\cost(X_r \setminus X_r^*,S) - \cost(Q_r,S)| 
      \le  \frac{1}{100} \cost(X_r\setminus X_r^*,S) 
     +  101 \cdot \cost(X_r\setminus X_r^*,S^*).
  \end{eqnarray*}
 Let $Q_{r,i}$ be the points in $Q_r$ that are nearest to center $c_i\in S$
 and let $X_{r,i}$ be the set of their original locations. For $p\in X_{r,i}$ that has been moved to $q_p \in Q_{r,i}$ with $q_p$ 
 let $c_i'$ be the closest point to $q_p$ not equal to $r$. Note that the only case in which $c_i\neq c_i'$ is when $c_i = r$. Furthermore, $q_p$ is not captured by $r$ because $r$ captures $c_r^*$ and is in $H$. So by \cref{proposition:dist} we know $\cost(\{q_p\},\{c_i'\})\leq \frac{\mu + 1}{\mu -1}\cost(\{q_p\},\{c_i\})$. Thus we have:
 \begin{eqnarray*}
   && |\cost(\{q_p\},\{c_i\}) - \cost(\{p\},\{c_i'\})| \\
   & = & |\cost(\{q_p\},\{c_i\})-\cost(\{q_p\},\{c_i'\})|
   +|\cost(\{q_p\},\{c_i'\}) - \cost(\{p\},\{c_i'\})| \\
   & \le &  \frac{2}{\mu - 1}\cost(\{q_p\},\{c_i\}) + \frac{1}{100} \cost(\{q_p\},\{c_i'\}) 
   + 101 \cdot  \cost(\{p\},\{q_p\}) \\
   & \le &  \frac{1}{\mu-1}\left(2+\frac{\mu+1}{100}\right) \cost(\{q_p\},\{c_i\}) 
   + 101 \cdot  \cost(\{p\},\{q_p\}),
 \end{eqnarray*}
 where in the first inequality we used  \cref{lemma:distances} with $\epsilon = 1/100$.
 
  Summing up over all points in $X_r\setminus X_r^*$ and the corresponding points in $Q_r$ yields
  \begin{eqnarray*}
    &&|\cost(Q_r,S) - \sum_{i} \cost(X_{r,i},S \setminus \{r\})|\\
    & \le&  \frac{1}{\mu-1}\left(2+\frac{\mu+1}{100}\right) \cost(Q_r,S) +
     101 \cdot  \cost(X_r\setminus X_r^*,S^*) \\
    & \le & \frac{26}{25} \bigg(\frac{11}{100} \cost(X_r\setminus X_r^*,S) + 
     \hspace{6mm} 101 \cdot \cost(X_r\setminus X_r^*,S^*)\bigg) + \\ 
    && 101 \cdot  \cost(X_r\setminus X_r^*,S^*) \\
    & \le & \frac{3}{25} \cost(X_r,S) + 231 \cdot  \cost(X_r,S^*),
  \end{eqnarray*}
  where the second inequality uses the bound on $\mu\geq 3$.
  Hence,
   \begin{eqnarray*}
  &&\reassign(X,S,c_r)  \\
  &&\hspace{4mm}=|\cost(X_r \setminus X_r^*,S) -
  \sum_{i}\cost(X_{r,i},S \setminus \{r\}|\\
  &&\hspace{4mm}\le |\cost(X_r \setminus X_r^* ,S) - \cost(Q_r,S)| + |\cost(Q_r,S) - \sum_{i} \cost(X_{r,i},S \setminus \{r\})|\\
   &&\hspace{4mm}\le \frac{13}{100} \cost(X_r,S) + 332 \cost(X_r,S^*). \qedhere
 \end{eqnarray*}
\end{proof}

\subsection{Proof of Lemma~\ref{lemma:good}}
\begin{proof}[Proof of Lemma~\ref{lemma:good}]
  We have
  $
  \sum_{h\in H} \cost(X_h^*,C) \ge \frac{1}{3} \cost(X,S)
  $
  and by the definition of good and \cref{lemma:reassign}
  \begin{eqnarray*}
  \sum_{h \in H, h \text{ is not good}} \cost(X_h^*,S) &\le& \sum_{h\in H} \reassign(X,S,c_h) + 
  9 \Opt_k + \frac{1}{100} \cost(X,S)\\
  & \le& \frac{14}{100} \cost(X,S) + 
     341 \Opt_k. 
  \end{eqnarray*}
  Using that $\cost(X,S)\ge 2000 \Opt_k$ we obtain that
  $$
  \sum_{h \in H, h \text{ is not good}} \cost(X_h^*,S) \le \frac{621}{2000} \cdot \cost(X,S).
  $$
  So $\sum_{h \in H, h \text{ is  good}} \cost(X_h^*,S) \geq \frac{9}{400} \cdot \cost(X,S)$. The lemma follows.
\end{proof}

\subsection{Proof of Lemma~\ref{lemma:good2}}
\begin{proof}[Proof of Lemma~\ref{lemma:good2}]
  We have
  $
  \sum_{r\in R} \cost(X_r^*,S) \ge 2/3 \cost(X,S).
  $
  Note that $|R| \le 2|L|$. By the definition of good and \cref{lemma:reassign}
  \begin{eqnarray*}
  \sum_{r \in R, r \text{ is not good}} &&\cost(X_r^*,S)  \\
  && \leq 2 |L| \min_{\ell\in L} \reassign(X,S,\ell) + 9 \Opt_k   +\frac{1}{100} \cost(X,S)\\
&&  \leq 2 \sum_{\ell\in L} \reassign(X,S,\ell) + 9 \Opt_k  + \frac{1}{100} \cost(X,S) \\
  && \leq\frac{27}{100} \cost(X,S) + 673 \Opt_k. 
  \end{eqnarray*}
  Using that $\sum_{i \in \{1,\dots,k\}} \cost(X_i^*,S)\ge 2000 \Opt_k$ we obtain that
 \begin{eqnarray*}
  \sum_{r \in R, r \text{ is not good}}\cost(X_r^*,S) &\leq&
  \frac{1213}{2000} \cost(X,S)
  \end{eqnarray*}
  Now the bound follows by combining the previous inequality with $$\sum_{r\in R} \cost(X_r^*,S) \ge 2/3 \cost(X,S)$$.
\end{proof}

\section{Further Theoretical Considerations}
\label{sec:moretheory}
Our algorithms can be used to obtain an $O(1)$-approximation for the case where $C \not\subseteq X$, using the following argument.
 Given an instance $A, C$ where $C \not\subseteq X$, consider running our algorithm on the instance $A, C'$ where $C' := A $ – hence looking at the metric induced by the points in $ A$ and falling back to the setting considered
 in our paper. This will give a solution whose centers are in $A$. We show that we can then transform this solution into a solution whose centers are in $C$ without losing much in the approximation guarantee. To do this, we need to show the following two statements: (1) the cost of the optimum solution in the instance $A,C'$ is at most $O(1)$ times the cost of the optimum solution in the instance $ A,C$; and (2) the cost of turning the solution for $A,C'$ into a solution for $A,C$ only loses a constant factor in the approximation guarantee.
 
(1) Take the optimum solution OPT for the instance $A,C$ and turn it into a solution for the instance $A,C'$ of cost at most $O(1)$ times higher. Replace each center in $\Opt$ with the closest element in $A$. This yields a solution for $A,C'$ (since $C' = A$). Note that by the triangle inequality, replacing each element in $\Opt$ with the closest point in $A$ only increases the cost by a factor 4. Thus, the cost of the optimum solution for $A,C'$ is only 4 times higher than the cost of $\Opt$. Hence, running our algorithm on the instance $A,C'$ yields a solution of cost $O(\Opt)$. We next show how to convert the solution obtained for $A,C'$ to a solution for $A,C$.

(2) Let's now show that we can transform any $\alpha$-approximate solution for the instance $A,C'$ to a solution for $A,C$ without losing more than a constant factor. Indeed, for each cluster $S$ of the solution for $A,C'$, pick the center $u_S$ in $C$ that is the closest to $S$ current center $c'$ in $A'$. For each cluster $S$, the cost obtained is by triangle inequality is at most $\sum_{s \in S} \dist(s, u_S) \le \sum_{s \in S} \dist(s, c') + \dist(c', u_S) \le \sum_{s \in S} \dist(s, c') + |S| \dist(c', u_S)$. Moreover, by the choice of $u_S$ and the triangle inequality, we have that $\dist(c', u_S) \le (1/|S|) \sum_{s \in S} (\dist(s, c') + \opt_s))$, where $\opt_s$ is the cost of $s$ in the optimum solution for $A,C$. Thus the overall cost is bounded by $\sum_{s \in S} (2\dist(s, c') + \opt_s)$. Summing over all clusters, we have that the total cost is at most $(18 \alpha + 1)$ times higher than the optimum cost for the instance $A,C$. Since we prove that our algorithm is an O(1)-approx, the resulting solution is an O(1)-approx too.

\end{document}